\documentclass[preprint,12pt,3p]{elsarticle}

\usepackage{natbib}

\usepackage{url}
\usepackage[colorlinks]{hyperref}
\hypersetup{
  citecolor = {blue}
}
\usepackage{amssymb}
\usepackage{amsmath}
\newtheorem{theorem}{Theorem}

\newenvironment{proof}{\noindent\textbf{Proof}.}{\hfill$\square$}
\newtheorem{definition}{Definition}
\newtheorem{example}{Example}
\usepackage{booktabs}

\usepackage{longtable}
\usepackage{multirow}
\usepackage{bm}
\usepackage[linesnumbered,ruled,vlined]{algorithm2e}

\usepackage{subcaption}
\usepackage{cleveref}

\usepackage{pgfplots}
\usepgfplotslibrary{dateplot,groupplots}
\pgfplotsset{compat=newest}

\usepackage{tikz}
\usetikzlibrary{bending,calc,chains,arrows.meta}
\usepackage{latexsym}

\allowdisplaybreaks

\DeclareMathOperator{\round}{round}

\makeatletter
\DeclareFontFamily{U}{tipa}{}
\DeclareFontShape{U}{tipa}{m}{n}{<->tipa10}{}
\newcommand{\arc@char}{{\usefont{U}{tipa}{m}{n}\symbol{62}}}%
\newcommand{\arc}[1]{\mathpalette\arc@arc{#1}}
\newcommand{\arc@arc}[2]{%
  \sbox0{$\m@th#1#2$}%
  \vbox{
    \hbox{\resizebox{\wd0}{\height}{\arc@char}}
    \nointerlineskip
    \box0
  }%
}
\makeatother

\journal{European Journal of Operational Research}

\begin{document}

\title{A case study of the profit-maximizing multi-vehicle pickup and delivery selection problem for the road networks with the integratable nodes}

\author[label1]{Aolong Zha\corref{cor1}}
\ead{a-zha@g.ecc.u-tokyo.ac.jp}
\address[label1]{Research Center for Advanced Science and Technology, The University of Tokyo}
\cortext[cor1]{Corresponding author}

\author[label2]{Qiong Chang}
\ead{q.chang@c.titech.ac.jp}
\address[label2]{School of Computing, Tokyo Institute of Technology}

\author[label1]{Naoto Imura}
\ead{nimura@g.ecc.u-tokyo.ac.jp}

\author[label1]{Katsuhiro Nishinari}
\ead{tknishi@mail.ecc.u-tokyo.ac.jp}

\begin{abstract}
This paper is a study of an application-based model in 
profit-maximizing multi-vehicle pickup and delivery selection problem (PPDSP). 
The graph-theoretic model proposed by existing studies of PPDSP is based on 
transport requests to define the corresponding nodes 
(i.e., each request corresponds to a pickup node and a delivery node). 
In practice, however, there are probably multiple requests coming from 
or going to an identical location. 
Considering the road networks with the integratable nodes as above, 
we define a new model based on the integrated nodes for the corresponding PPDSP 
and propose a novel mixed-integer formulation. 
In comparative experiments with the existing formulation, 
as the number of integratable nodes increases, 
our method has a clear advantage in terms of the number of variables 
as well as the number of constraints required in the generated instances, 
and the accuracy of the optimized solution obtained within a given time. 
\end{abstract}


\maketitle

\section{Introduction}

Recently, 
the profit-maximizing multi-vehicle pickup and delivery selection problem (PPDSP) 
has gained a lot of attention in the field of practical transportation and logistics 
\cite{conf/cpaior/LiuAB18,journals/cor/RiedlerR18,ASGHARI2020101815,HUANG20201}. 
This problem was first proposed in \cite{journals/eor/QiuFN17}, 
which involves three classical problem models: 
\emph{routing optimization}, \emph{pickup and delivery}, and 
\emph{selective pickup} (a.k.a. \emph{knapsack}). 
Solving this problem quickly and optimally can both help 
improve the operational efficiency of the carriers 
and contributes to more eco-friendly transportation. 

To the best of our knowledge, in the graph-theoretic models constructed 
in the existing studies 
\cite{journals/isci/TingLHL17,journals/ors/GanstererKH17,conf/gol/Al-ChamiFMM18,journals/eor/Ahmadi-JavidAM18}, the definitions of nodes are based on 
the pickup and delivery location from the requests. 
This means that one request needs to correspond to two nodes. 
However, in application scenarios, 
there are often plural requests coming from or to arrive at the same location. 
The number of variables to be required in the existing mixed-integer formulation 
heavily depends on the number of nodes in the model, 
and since PPDSP is an $\mathcal{NP}$-hard problem, 
its computational complexity is exponential with respect to the number of variables. 

\begin{figure}
\centering
\makebox[\textwidth]{
\scalebox{.5}{
\begin{subfigure}{\textwidth}
\centering
\begin{tikzpicture}[scale=1]
\begin{scope}[auto, every node/.style={draw, thick, minimum height=65pt, rounded corners=5pt, align=left, font=\Large}]
    \node [inner sep=16pt] (1) at (0,0) {~~~~Depot~~~~};
    \node [inner sep=3pt] (2) at (5,0) {Req. 1\\pickup\\coord.};
    \node [inner sep=3pt] (3) at (7.15,0) {Req. 2\\pickup\\coord.};
    \node [inner sep=3pt] (4) at (5,-4) {Req. 3\\pickup\\coord.};
    \node [inner sep=3pt] (5) at (7.15,-4) {Req. 2\\dropoff\\coord.};
    \node [inner sep=3pt] (6) at (-1.1,-4) {Req. 3\\dropoff\\coord.};
    \node [inner sep=3pt] (7) at (1.1,-4) {Req. 1\\dropoff\\coord.};

    \path [->, line width=2pt] (1) edge node[draw=none] {} (2);
    \path [->, line width=2pt] (2) edge node[draw=none] {} (3);
    \path [->, line width=2pt] (3) edge node[draw=none] {} (5);
    \path [->, line width=2pt] (5) edge node[draw=none] {} (4);
    \path [->, line width=2pt] (4) edge node[draw=none] {} (7);
    \path [->, line width=2pt] (7) edge node[draw=none] {} (6);
    \path [->, line width=2pt] (6) edge node[draw=none] {} (1);
\end{scope}
\end{tikzpicture}
\captionsetup{font=Large}
\caption{}
\label{fig:1.1}
\end{subfigure}
\hspace{-150pt}
\begin{subfigure}{\textwidth}
\centering
\begin{tikzpicture}[scale=1]
\begin{scope}[auto, every node/.style={draw, thick, minimum height=65pt, rounded corners=5pt, align=left, font=\Large}]
    \node [inner sep=16pt] (1) at (0,0) {~~~~Depot~~~~};
    \node [inner sep=21pt] (2) at (6,0) {Location 1};
    \node [inner sep=21pt] (3) at (6,-4) {Location 2};
    \node [inner sep=21pt] (4) at (0,-4) {Location 3};

    \path [->, line width=2pt] (1) edge node[draw=none] {} (2);
    \path [->, line width=2pt] (2) edge node[draw=none] {} (3);
    \path [->, line width=2pt] (3) edge node[draw=none] {} (4);
    \path [->, line width=2pt] (4) edge node[draw=none] {} (1);
\end{scope}
\end{tikzpicture}
\captionsetup{font=Large}
\caption{}
\label{fig:1.2}
\end{subfigure}
\hspace{-150pt}
\begin{subfigure}{\textwidth}
\centering
\begin{tikzpicture}[scale=1]
\begin{scope}[auto, every node/.style={draw, thick, minimum height=65pt, rounded corners=5pt, align=left, font=\Large}]
    \node [inner sep=16pt] (1) at (0,0) {~~~~Depot~~~~};
    \node [inner sep=3pt] (2) at (5,0) {Req. 1\\pickup\\coord.};
    \node [inner sep=3pt] (3) at (7.15,0) {Req. 2\\pickup\\coord.};
    \node [inner sep=3pt] (4) at (5,-4) {Req. 3\\pickup\\coord.};
    \node [inner sep=3pt] (5) at (7.15,-4) {Req. 2\\dropoff\\coord.};
    \node [inner sep=3pt] (6) at (-1.1,-4) {Req. 3\\dropoff\\coord.};
    \node [inner sep=3pt] (7) at (1.1,-4) {Req. 1\\dropoff\\coord.};

    \path [->, line width=2pt] (1) edge node[draw=none] {} (2);
    \path [->, line width=2pt] (2) edge node[draw=none] {} (7);
    \path [->, line width=2pt] (7) edge node[draw=none] {} (3);
    \path [->, line width=2pt] (3) edge node[draw=none] {} (5);
    \path [->, line width=2pt] (5) edge node[draw=none] {} (4);
    \path [->, line width=2pt] (4) edge [out=220,in=320] node[draw=none] {} (6);
    \path [->, line width=2pt] (6) edge node[draw=none] {} (1);
\end{scope}
\end{tikzpicture}
\vspace*{-45pt}
\captionsetup{font=Large}
\caption{}
\label{fig:1.3}
\end{subfigure}
}
}
\caption{A simple example for explaining the relationship between 
the request-based model (i.e., Figures \ref{fig:1.1} and \ref{fig:1.3}) and 
the location-based model (i.e., \Cref{fig:1.2}), where 
\{Req. 1 pickup coord.\} and \{Req. 2 pickup coord.\}, 
\{Req. 3 pickup coord.\} and \{Req. 2 dropoff coord.\}, and 
\{Req. 3 dropoff coord.\} and \{Req. 1 dropoff coord.\} can be integrated as 
\{location 1\}, \{location 2\} and \{location 3\}, respectively. 
The route shown in \Cref{fig:1.3} can be regarded as a feasible solution 
of the request-based model, but cannot be corresponded to in 
the location-based model.}
\label{fig:1}
\end{figure}
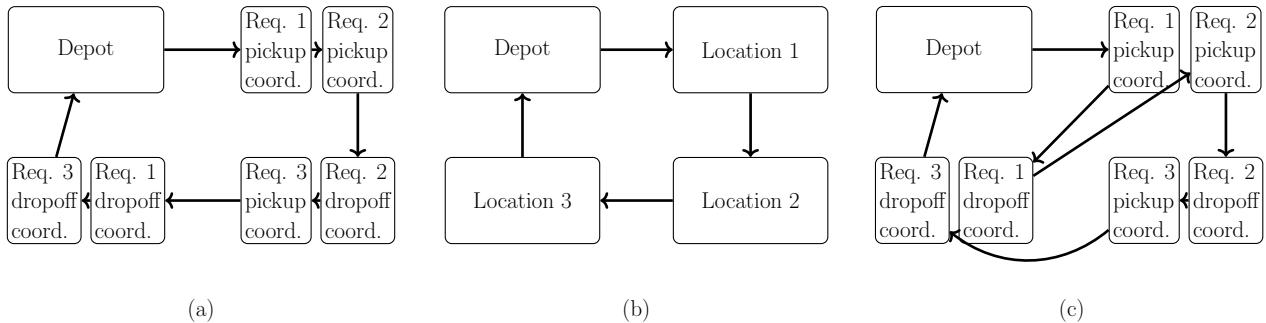

The motivation of this study is to provide a more reasonable and effective 
mathematical model for practical logistics and transportation problems. 
Considering the road networks with the integratable nodes as above, 
we proposes a new concise model in which the definition of nodes 
is based on the locations rather than the requests, 
and give a novel mixed-integer programming formulation 
that reduces the number of required variables. 
It is necessary to claim that the solution set of 
our proposed location-based method is a subset of 
the solution set of the request-based method. 
This is because the integrated node can only be passed 
at most once under the \emph{Hamiltonian cycle} constraint, 
where ``at most'' is due to the objective function of 
profit-maximizing (i.e., it is possible that any node 
will not be traversed). 
In \Cref{fig:1}, we take the delivery of single truck as an example, 
and assume that the truck is transported according to the route 
given in \Cref{fig:1.1} as a feasible solution in this scenario, 
where the pickup coordinates of request 1 (abbr. Req. 1 pickup coord.) 
and Req. 2 pickup coord. are the identical locations. 
We can integrate them as location 1 as shown in \Cref{fig:1.2}, 
and after integrating all the same coordinates into one location 
respectively, our model becomes more concise. 
This integration changes the nodes of the model from 
repeatable coordinates based on the request to the unique locations, 
which results in each location can only be visited once. 
Therefore, the location-based model can also correspond to 
the feasible solutions in \Cref{fig:1.1}, but not to some of 
the feasible solutions corresponded to the request-based model, 
such as the solution given in \Cref{fig:1.3}. 
Whereas in real logistics application scenarios, 
making multiple retraces to the same location is rare 
(e.g. insufficient vehicle capacity, goods cannot be mixed, etc.), 
and also non-efficient. 
For such locations with a large number of requests or high loading volume, 
it is a more common strategy to assign multiple vehicles to them. 
Therefore, we argue that the location-based model 
can improve the optimization efficiency although 
it reduces the range of feasible solutions.

\section{Preliminaries}
\label{sec:2}

Given a directed graph $G=(V,E)$, where 
$V=\{0,1,2,\ldots,|V|\}$ is the set of nodes 
representing the location points (0 is depot) 
and $E$ is the set of arc, denoted as $\arc{od}$, 
and given a set of trucks $T=\{1,2,\ldots,|T|\}$, 
we define a type of Boolean variables $x^t_{od}$ that 
is equal to 1 if the truck $t$ passes through 
the arc $\arc{od}$ and 0 otherwise. 
We denote the load capacity of truck $t$ as $c^t$, and 
the cost of the truck $t$ traversing 
the arc $\arc{od}$ as $l^t_{od}$, 
where $t\in T$ and $o,d\in V$. 

Let $R=\{1,2,\ldots,|R|\}$ be a set of requests. 
Each request $r$ $(r\in R)$ is considered as a tuple 
$r=\langle w_r, q_r, f(r), g(r)\rangle$, where 
\begin{itemize}
\item $w_r$ is the the payment that can be received 
for completing the shipping of request $r$; 
\item $q_r$ is the volume of request $r$; 
\item $f(r)$ is the loading point of request $r$; 
\item $g(r)$ is the unloading point of request $r$, 
\end{itemize}
and $f:R\to V\setminus\{0\}$ (resp. $g:R\to V\setminus\{0\}$) 
is a function for mapping the loading (resp. unloading) 
point of requests $r$. 
We also define another type of Boolean variables $y^t_r$ 
that is equal to 1 if request $r$ is allocated to truck $t$ 
and 0 otherwise, where $r\in R$ and $t\in T$. 
Besides, we denote the number of location points visited 
by truck $t$ before it reaches $v~(v\in V\setminus\{0\})$ as $u^t_v$. 

\begin{definition}[Delivery of Truck $t$]
Delivery of truck $t$ is denoted by 
$D_t=\bigcup_{r\in R}\{r\,|\,y^t_r=1\}$, 
where $t\in T$. 
\end{definition}

\begin{definition}[Route of Truck $t$]
Route of truck $t$ is denoted by 
$S_t=\bigcup_{o\in V}\bigcup_{d\in V}\{\arc{od}\,|\,x^t_{od}=1\}$, 
where $t\in T$. $S_t$ satisfies the following conditions 
if $D_t\neq\emptyset$: 
\begin{description}
\item [Hamiltonian Cycle] 
\begin{itemize}
\item Denote $P_t=\{0\}\cup\bigcup_{r\in D_t}\{f(r),g(r)\}$ as the set of location points visited and departed exactly once by truck $t$, where the predecessor and successor are the same node is not counted (e.g., even if $\arc{vv}\,|\,x^t_{vv}=1$, neither this time can be included in the number of visits or departures of truck $t$ to/from location $v$.); 
\item Ensure that no subtour exists in $S_t$.\footnote{
    Existing studies usually include the constraints on time windows, which can eliminate subtour. In this study, in order to compare the request-based model with the location-based model, common parts of both models are omitted (e.g., the time windows constraints). Instead, we use the most basic MTZ-formulation to eliminate subtour.
}
\end{itemize}
\item [Loading Before Unloading] $\forall i\in D_t,u^t_{f(i)}<u^t_{g(i)}.$
\item [Capacity Limitation] At any time, the total volume of cargo carried by truck $t$ cannot exceed its capacity $c^t$. 
\end{description}
\end{definition}

\begin{definition}[Delivery Routing Solution]
Delivery routing solution $DS=\bigcup_{t\in T}\{(D_t,S_t)\}$ 
that is a partition of $R$ into disjoint and contained $D_t$ with 
the corresponding $S_t$:
\begin{align*}
    &\forall i,j\,(i\neq j),~D_i\cap D_j=\emptyset,~\bigcup_{D_i\in DS}D_i\subseteq R,~\bigcup_{S_i\in DS}S_i\subseteq E.
\end{align*}
Denote the set of all possible delivery routing solutions as $\Pi(R,T)$. 
\end{definition}

\begin{definition}[Profit-Cost Function]
A profit-cost function assigns a real-valued profit to every $D_t$: 
$w:D_t\to\mathbb{R}$ and a cost to every $S_t$: 
$l:S_t\to\mathbb{R}$, where $w(D_t)=\sum_{r\in D_t}w_r$ and 
$l(S_t)=\sum_{\arc{od}\in S_t}l^t_{od}$. 
For any delivery routing solution $DS\in\Pi(R,T)$, 
the value of $DS$ is calculated by 
\begin{align*}
    &\xi(DS)=\sum_{D_t\in DS}w(D_t)-\sum_{S_t\in DS}l(S_t).
\end{align*}
\end{definition}

In general, a delivery routing solution $DS$ that 
considers only maximizing profits or minimizing costs 
is not necessarily an optimal $DS$. 
Therefore, we have to find the optimal delivery routing solution 
that maximizes the sum of the values of profit-cost functions. 
We define a delivery routing problem in profit-cost function 
as follows. 

\begin{definition}[PPDSP]
For a set of requests and trucks $(R,T)$, a profit-maximizing 
multi-vehicle pickup and delivery selection problem (PPDSP) 
is to find the optimal delivery routing solution $DS^*$ such that 
\begin{align*}
    DS^*\in\arg\max_{DS\in\Pi(R,T)}\xi(DS).
\end{align*}
\end{definition}

Here we show an example of PPDSP. 

\begin{example}
Assume that two trucks $T=\{t_1,t_2\}$ are responsible for 
three requests $R=\{r_1,r_2,r_3\}$ 
with the following conditions. 
\begin{itemize}
    \item The information about the requests:\\[8pt]
        \begin{tabular}{ccccc}
        \toprule
        Request & $w_r$ & $q_r$ & $f(r)$ & $g(r)$\\
        \midrule
        $r_1$ & $13$ & $4$ & $a$ & $c$\\
        $r_2$ &  $7$ & $2$ & $a$ & $b$\\
        $r_3$ &  $4$ & $1$ & $b$ & $c$\\
        \bottomrule
        \end{tabular}
    \item The information about the trucks:
        \vspace{-8pt}
        \begin{itemize}
            \item The capacities of the trucks $c^{t_1}=6$ 
            and $c^{t_2}=3$. 
            \item The cost matrices of each truck through each arc 
            ($\delta$ is depot).\\[8pt]
            \begin{tabular}{|c|c|c|c|c|c|}
            \hline
            $l^{t_1}_{od}$ & $\delta$ & $a$ & $b$ & $c$\\
            \hline
            $\delta$ & $0$ & $2$ & $2$ & $2$\\
            \hline
            $a$      & $2$ & $0$ & $4$ & $7$\\
            \hline
            $b$      & $2$ & $4$ & $0$ & $2$\\
            \hline
            $c$      & $2$ & $7$ & $2$ & $0$\\
            \hline
            \end{tabular}
            \hspace{10pt}
            \begin{tabular}{|c|c|c|c|c|c|}
            \hline
            $l^{t_2}_{od}$ & $\delta$ & $a$ & $b$ & $c$\\
            \hline
            $\delta$ & $0$ & $1$ & $1$ & $1$\\
            \hline
            $a$      & $1$ & $0$ & $3$ & $5$\\
            \hline
            $b$      & $1$ & $3$ & $0$ & $1$\\
            \hline
            $c$      & $1$ & $5$ & $1$ & $0$\\
            \hline
            \end{tabular}
        \end{itemize}
\end{itemize}
For all $DS\in\Pi(R,T)$, we have their respective profit-cost function as follows. 
\begin{longtable}{l}
$\xi(\{(D_{t_1}=\emptyset, 
S_{t_1}=\emptyset), 
(D_{t_2}=\emptyset, 
S_{t_2}=\emptyset)\})=0$,\\
$\xi(\{(D_{t_1}=\emptyset, 
S_{t_1}=\emptyset), 
(D_{t_2}=\{r_2\}, 
S_{t_2}=\{\arc{\delta a},\arc{ab},\arc{b\delta}\})\})=2$,\\
$\xi(\{(D_{t_1}=\emptyset, 
S_{t_1}=\emptyset), 
(D_{t_2}=\{r_3\}, 
S_{t_2}=\{\arc{\delta b},\arc{bc},\arc{c\delta}\})\})=1$,\\
$\xi(\{(D_{t_1}=\emptyset, 
S_{t_1}=\emptyset), 
(D_{t_2}=\{r_2,r_3\}, 
S_{t_2}=\{\arc{\delta a},\arc{ab},\arc{bc},\arc{c\delta}\})\})=5$,\\
$\xi(\{(D_{t_1}=\{r_1\}, 
S_{t_1}=\{\arc{\delta a},\arc{ac},\arc{c\delta}\}), 
(D_{t_2}=\emptyset, 
S_{t_2}=\emptyset)\})=2$,\\
$\xi(\{(D_{t_1}=\{r_1\}, 
S_{t_1}=\{\arc{\delta a},\arc{ac},\arc{c\delta}\}), 
(D_{t_2}=\{r_2\}, 
S_{t_2}=\{\arc{\delta a},\arc{ab},\arc{b\delta}\})\})=4$,\\
$\xi(\{(D_{t_1}=\{r_1\}, 
S_{t_1}=\{\arc{\delta a},\arc{ac},\arc{c\delta}\}), 
(D_{t_2}=\{r_3\}, 
S_{t_2}=\{\arc{\delta b},\arc{bc},\arc{c\delta}\})\})=3$,\\
$\xi(\{(D_{t_1}=\{r_1\}, 
S_{t_1}=\{\arc{\delta a},\arc{ac},\arc{c\delta}\}), 
(D_{t_2}=\{r_2,r_3\}, 
S_{t_2}=\{\arc{\delta a},\arc{ab},\arc{bc},\arc{c\delta}\})\})=7$,\\
$\xi(\{(D_{t_1}=\{r_2\}, 
S_{t_1}=\{\arc{\delta a},\arc{ab},\arc{b\delta}\}), 
(D_{t_2}=\emptyset, 
S_{t_2}=\emptyset)\})=-1$,\\
$\xi(\{(D_{t_1}=\{r_2\}, 
S_{t_1}=\{\arc{\delta a},\arc{ab},\arc{b\delta}\}), 
(D_{t_2}=\{r_3\}, 
S_{t_2}=\{\arc{\delta b},\arc{bc},\arc{c\delta}\})\})=0$,\\
$\xi(\{(D_{t_1}=\{r_3\}, 
S_{t_1}=\{\arc{\delta b},\arc{bc},\arc{c\delta}\}), 
(D_{t_2}=\emptyset, 
S_{t_2}=\emptyset)\})=-2$,\\
$\xi(\{(D_{t_1}=\{r_3\}, 
S_{t_1}=\{\arc{\delta b},\arc{bc},\arc{c\delta}\}), 
(D_{t_2}=\{r_2\}, 
S_{t_2}=\{\arc{\delta a},\arc{ab},\arc{b\delta}\})\})=0$,\\
$\xi(\{(D_{t_1}=\{r_1,r_2\}, 
S_{t_1}=\{\arc{\delta a},\arc{ab},\arc{bc},\arc{c\delta}\}), 
(D_{t_2}=\emptyset, 
S_{t_2}=\emptyset)\})=10$,\\
$\xi(\{(D_{t_1}=\{r_1,r_2\}, 
S_{t_1}=\{\arc{\delta a},\arc{ac},\arc{cb},\arc{b\delta}\}), 
(D_{t_2}=\emptyset, 
S_{t_2}=\emptyset)\})=7$,\\
$\xi(\{(D_{t_1}=\{r_1,r_2\}, 
S_{t_1}=\{\arc{\delta a},\arc{ab},\arc{bc},\arc{c\delta}\}), 
(D_{t_2}=\{r_3\}, 
S_{t_2}=\{\arc{\delta b},\arc{bc},\arc{c\delta}\})\})=11$,\\
$\xi(\{(D_{t_1}=\{r_1,r_2\}, 
S_{t_1}=\{\arc{\delta a},\arc{ac},\arc{cb},\arc{b\delta}\}), 
(D_{t_2}=\{r_3\}, 
S_{t_2}=\{\arc{\delta b},\arc{bc},\arc{c\delta}\})\})=8$,\\
$\xi(\{(D_{t_1}=\{r_1,r_3\}, 
S_{t_1}=\{\arc{\delta a},\arc{ab},\arc{bc},\arc{c\delta}\}), 
(D_{t_2}=\emptyset, 
S_{t_2}=\emptyset)\})=7$,\\
$\xi(\{(D_{t_1}=\{r_1,r_3\}, 
S_{t_1}=\{\arc{\delta b},\arc{ba},\arc{ac},\arc{c\delta}\}), 
(D_{t_2}=\emptyset, 
S_{t_2}=\emptyset)\})=2$,\\
$\xi(\{(D_{t_1}=\{r_1,r_3\}, 
S_{t_1}=\{\arc{\delta a},\arc{ab},\arc{bc},\arc{c\delta}\}), 
(D_{t_2}=\{r_2\}, 
S_{t_2}=\{\arc{\delta a},\arc{ab},\arc{b\delta}\})\})=9$,\\
$\xi(\{(D_{t_1}=\{r_1,r_3\}, 
S_{t_1}=\{\arc{\delta b},\arc{ba},\arc{ac},\arc{c\delta}\}), 
(D_{t_2}=\{r_2\}, 
S_{t_2}=\{\arc{\delta a},\arc{ab},\arc{b\delta}\})\})=4$,\\
$\xi(\{(D_{t_1}=\{r_2,r_3\}, 
S_{t_1}=\{\arc{\delta a},\arc{ab},\arc{bc},\arc{c\delta}\}), 
(D_{t_2}=\emptyset, 
S_{t_2}=\emptyset)\})=1$.
\end{longtable}
In this example, the optimal delivery routing solution $DS^*$ 
for the given PPDSP is $\{(D_{t_1}=\{r_1,r_2\}, 
S_{t_1}=\{\arc{\delta a},\arc{ab},\arc{bc},\arc{c\delta}\}), 
(D_{t_2}=\{r_3\}, 
S_{t_2}=\{\arc{\delta b},\arc{bc},\arc{c\delta}\})\}$, and 
its value is $11$. 
\end{example}

\section{Problem Formulation}
\label{sec:3}

In this section, we present the following 
mixed-integer programming (MIP) formulation of 
the location-based model for PPDSP and 
prove its correctness as well as 
the space complexity of the generated problem. 

\begin{align}
\max~&\sum_{r\in R}\sum_{t\in T}(w_{r}\cdot y^t_r)-\sum_{t\in T}\sum_{o\in V}\sum_{d\in V}(l^t_{od}\cdot x^t_{od}),\label{eq:1}\\
\textrm{s.t.}~&x^t_{od},y^t_r\in\{0,1\},&\hspace{-160pt}\forall (t,o,d,r):t\in T, o,d\in V, r\in R,\label{eq:2}\\
&\sum_{t\in T}y^t_r\le 1,&\hspace{-160pt}\forall r:r\in R,\label{eq:3}\\
&y^t_r\le\sum_{o\in V\atop o\neq f(r)}x^t_{of(r)},&\hspace{-160pt}\forall (t,r):t\in T, r\in R,\label{eq:4}\\
&y^t_r\le\sum_{o\in V\atop o\neq g(r)}x^t_{og(r)},&\hspace{-160pt}\forall (t,r):t\in T, r\in R,\label{eq:5}\\
&\sum_{d\in V}x^t_{od}-\sum_{d\in V}x^t_{do}=0,&\hspace{-160pt}\forall (t,o):t\in T, o\in V,\label{eq:6}\\
&\sum_{d\in V\atop d\neq o}x^t_{od}\le 1,&\hspace{-160pt}\forall (t,o):t\in T, o\in V,\label{eq:7}\\
&u^t_d-u^t_o\ge 1-|V|(1-x^t_{od}),\nonumber\\
&&\hspace{-160pt}\forall (t,o,d):t\in T, o,d\in V\setminus\{0\}, o\neq d,\label{eq:8}\\
&u^t_{f(r)}-u^t_{g(r)}<|V|(1-y^t_r),&\hspace{-160pt}\forall (t,r):t\in T, r\in R,\label{eq:9}\\
&\begin{cases}
\Gamma-M\cdot(1-x^t_{od})\le h^t_d-h^t_o\le\Gamma+M\cdot(1-x^t_{od}),\\
\Gamma\triangleq\sum_{r\in f^{-1}(d)}(q_r\cdot y^t_r)-\sum_{r\in g^{-1}(d)}(q_r\cdot y^t_r),\\
M\triangleq c^t+\sum_{r\in R}q_r,
\end{cases},\nonumber\\
&&\hspace{-160pt}\forall (t,o,d):t\in T, o,d\in V\setminus\{0\}, o\neq d,\label{eq:10}\\
&0\le h^t_v\le c^t,&\hspace{-160pt}\forall (t,v):t\in T, v\in V\setminus\{0\},\label{eq:11}\\
&0\le u^t_v\le |V|-2,&\hspace{-160pt}\forall (t,v):t\in T, v\in V\setminus\{0\}.\label{eq:12}
\end{align}

\subsection{Correctness}
\label{sec:3.1}

The objective function in Eq. \eqref{eq:1} maximizes the profit-cost function 
for all possible delivery routing solutions. 
Eq. \eqref{eq:2} introduces two types of Boolean variables $x^t_{od}$ and $y^t_r$. 
Eq. \eqref{eq:3} guarantees that, each request can be assigned to at most one truck. 

\begin{theorem}
\label{thm:1}
For PPDSP, the Hamiltonian cycle constraints can be guaranteed by the simultaneous 
Eqs. \eqref{eq:4}--\eqref{eq:8} and Eq. \eqref{eq:12}. 
\end{theorem}
\begin{proof}
According to Eq. \eqref{eq:4} (resp. Eq. \eqref{eq:5}), 
if request $r$ is assigned to truck $t$, then truck $t$ 
reaches the pickup (resp. dropoff) point of $r$ at least once via 
a point that is not the pickup (resp. dropoff) point of $r$. 
Eq. \eqref{eq:6} ensures that the number of visits and 
the number of departures of a truck at any location must be equal. 
Eq. \eqref{eq:7} restricts any truck departing from 
a location to at most one other location. 
It is clear from the above that, all nodes of $\bigcup_{r\in D_t}\{f(r),g(r)\}$ 
are ensured to be visited and departed by truck $t$ exactly once. 

Eq. \eqref{eq:8}, a canonical MTZ subtour elimination constraint \cite{journals/jacm/MillerTZ60}, 
resrticts the integer variables $u^t_v$, where $v\in V\setminus\{0\}$, 
such that they form an ascending series to represent the order 
that truck $t$ arrives at each location $v$. 
Eq. \eqref{eq:12} gives the domain of such variables $u^t_v$. 
Furthermore, Eq. \eqref{eq:8} associated with Eqs. \eqref{eq:4}--\eqref{eq:7} 
also enforces that depot 0 must be visited and departed by 
truck $t$ exactly once if $D_t\neq\emptyset$. 
\end{proof}

Eq. \eqref{eq:9} ensures that, for any request $r$, the arrival of its 
pickup point must precede the arrival of its dropoff point by truck $t$ 
if $y^t_r=1$. 

\begin{theorem}
\label{thm:2}
For PPDSP, the capacity constraint can be guaranteed by the 
Eqs. \eqref{eq:10} and \eqref{eq:11}. 
\end{theorem}
\begin{proof}
We introduce another integer variables $h^t_v$, where $v\in V\setminus\{0\}$, 
in Eq. \eqref{eq:10}, whose domains are between 0 and $c^t$ 
(i.e., the capacity of truck $t$) as given in Eq. \eqref{eq:11}. 
Such a variable $h^t_v$ can be considered as the loading volume of truck $t$ 
when it departs at location $v$. 
We define the amount of change in the loading of truck $t$ 
at location $v$ as $\Gamma$, which equals the loaded amount at $v$ 
(i.e., $\sum_{r\in f^{-1}(v)}(q_r\cdot y^t_r)$)
minus the unloaded amount at $v$ (i.e., $\sum_{r\in g^{-1}(v)}(q_r\cdot y^t_r)$). 
The main part of Eq. \eqref{eq:10} guarantees that $h^t_d-h^t_o$ is exactly 
equal to $\Gamma$ when $x^t_{od}=1$, which is written as a big-$M$ linear inequality, 
where we specify $M$ as $c^t+\sum_{r\in R}q_r$. 
\end{proof}

\subsection{Space Complexity}
\label{sec:3.2}

Consider that the number of trucks is $m$ (i.e., $|T|=m$) and 
the number of requests is $n$ (i.e., $|R|=n$). 
Assume that the average repetition rate of the same locations is $k$, where $k\ge 1$, 
then we can estimate the number of unique locations as $2nk^{-1}$, 
which is also the number of integrated nodes (viz., $|V\setminus\{0\}|=2nk^{-1}$). 

\begin{theorem}
\label{thm:3}
The number of linear (in)equations corresponding to 
Eqs. \eqref{eq:3}--\eqref{eq:10} is always of a pseudo-polynomial size, 
and both this number and the number of required variables are 
bounded by $\Theta(mn^2k^{-2})$. 
\end{theorem}
\begin{proof}
The number of Boolean variables $x^t_{od}$ (resp. $y^t_r$) required 
in proposed formulation is $m(1+2nk^{-1})^2$ (resp. $mn$); 
while both the number of required integer variables $u^t_d$ and $h^t_d$ 
are $m(1+2nk^{-1})$. 
In \Cref{tab:3.1}, we list the bounded numbers of linear (in)equations 
that correspond to Eqs. \eqref{eq:3}--\eqref{eq:10} 
involved in the proposed PPDSP formulation. 
Therefore, the number of required variable as well as the number of 
corresponded linear (in)equations are of $\Theta(mn^2k^{-2})$. 
\end{proof}

\setcounter{table}{0}
\begin{table}[t]
\caption{The bounded numbers of linear (in)equations corresponding to 
the constraints that are formulated in Eqs. \eqref{eq:3}--\eqref{eq:10}.}
\label{tab:3.1}
\centering
\begin{tabular}{lc|lc}
\toprule
Constraint & $\#$(In)equations & Constraint & $\#$(In)equations \\
\midrule
Eq. \eqref{eq:3} & $n$             & Eq. \eqref{eq:7}  & $m(1+2nk^{-1})$ \\
Eq. \eqref{eq:4} & $mn$            & Eq. \eqref{eq:8}  & $m\binom{2nk^{-1}}{2}$ \\
Eq. \eqref{eq:5} & $mn$            & Eq. \eqref{eq:9}  & $mn$ \\
Eq. \eqref{eq:6} & $m(1+2nk^{-1})$ & Eq. \eqref{eq:10} & $m\binom{2nk^{-1}}{2}$ \\
\bottomrule
\end{tabular}
\end{table}

\section{Experiment}
\label{sec:4}

We compare the performance of a MIP optimizer 
in solving randomly generated PPDSP instances based on the proposed formulation 
(i.e., location-based model), and based on the existing formulation 
(i.e., request-based model), respectively. 
Please refer to \Cref{apdx:1} for the specific description of 
the existing formulation used for the comparative experiments. 

\subsection{Instances Generation}
\label{sec:4.1}

The directed graph informations for generating instances 
are set based on the samples of TSPLIB benchmark with 
displaying data as the coordinates of nodes.\footnote{
	\url{http://comopt.ifi.uni-heidelberg.de/software/TSPLIB95/tsp/}
} 
We denote the number of nodes contained in the selected sample as $|V|$, 
and set the first of these nodes to be the depot node. 
The pickup and dropoff points of each request are chosen from 
the non-depot nodes. 
Therefore, given an average repetition rate $k$ for the same locations, 
we need to repeatably select $n$ pairs of non-depot nodes 
(i.e., a total of $2n$ repeatable non-depot nodes) to correspond to $n$ requests, 
where $n=\round(\frac{k(|V|-1)}{2})$. 

In \Cref{alg:4.1}, we construct the \textsl{repeaTimeList} of length $|V|-1$ 
to mark the number of times of each non-depot node being selected. 
In order to ensure that each non-depot node is selected at least once, 
each value on such list is initialized to 1. 
We continuously generate a random index of \textsl{repeaTimeList} and 
add one to the value corresponding to the generated index 
(i.e., the number of times of the non-depot node being selected increases by one) 
until the sum of these numbers reaches $2n$. 

\begin{algorithm}[h!]
\DontPrintSemicolon
\caption{Randomly specify the number of repetitions of each non-depot node}
\label{alg:4.1}
\SetDataSty{texttt}
\SetCommentSty{textrm}
\SetKwComment{tcp}{$\rhd$ }{}
\SetKwData{true}{true}
\SetKwData{false}{false}
\SetKwInOut{inpt}{Input}
\SetKwInOut{init}{Init.}
\inpt{$G=(V,E)$, $k$}
\init{$n\leftarrow\round(\frac{k(|V|-1)}{2})$, $\textsl{repeaTimeList}\leftarrow[|V|-1~\textrm{of}~1]$}
\While{$\sum_i \textsl{repeaTimeList}[i]<2n$}{
	$i\leftarrow\round(\textsc{RandomUniform}(0,|V|-2))$\;
	$\textsl{repeaTimeList}[i]\leftarrow\textsl{repeaTimeList}[i]+1$\;
}
\Return{\textsl{repeaTimeList}}\;
\end{algorithm}

The pickup and dropoff nodes for each request are randomly paired up in 
\Cref{alg:4.2}. We first rewrite \textsl{repeaTimeList} as a list of length $2n$, 
\textsl{shuffList}, which consists of all selected non-depot nodes, 
where the number of repetitions indicates the number of times they are selected. 
For example, we have $\textsl{shuffList}=[0,0,0,1,2,2]$ for 
$\textsl{repeaTimeList}=[3,1,2]$. Next, we shuffle \textsl{shuffList} 
and clear \textsl{pairList}, which is used to store pairs of nodes 
indicating the pickup and dropoff points of all requests. 
We then divide \textsl{shuffList} into $n$ pairs in order. 
If two nodes of any pair are identical (i.e., the pickup and dropoff 
are the same point), or the pair with considering order 
is already contained in \textsl{pairList}, 
then we are back to \Cref{alg:4.2:line:4}. 
Otherwise, we append the eligible pair of nodes into \textsl{pairList} 
until the last pair is added. 

\begin{algorithm}[h!]
\DontPrintSemicolon
\caption{Randomly pair up the pickup and dropoff nodes for each request}
\label{alg:4.2}
\SetDataSty{texttt}
\SetCommentSty{textrm}
\SetKwComment{tcp}{$\rhd$ }{}
\SetKwData{true}{true}
\SetKwData{false}{false}
\SetKwInOut{inpt}{Input}
\SetKwInOut{init}{Init.}
\inpt{$G=(V,E)$, $n=\round(\frac{k(|V|-1)}{2})$, \textsl{repeaTimeList}}
\init{$\textsl{shuffList}\leftarrow[\,]$, $\textsl{pairList}\leftarrow[\,]$, $\textsl{reshuffle}\leftarrow\true$}
\For{$i\leftarrow 0$ \KwTo $|V|-2$}{
	$\textsl{shuffList}.\textsc{Extend}([\textsl{repeaTimeList}[i]~\textrm{of}~i])$\;
}
\While{\textsl{reshuffle}}{
	\textsc{RandomShuffle}(\textsl{shuffList})\;\label{alg:4.2:line:4}
	$\textsl{pairList}\leftarrow[\,]$\;
	\For{$i\leftarrow 0$ \KwTo $n-1$}{
		\eIf{$\textsl{shuffList}[2i]=\textsl{shuffList}[1+2i]$~\ensuremath{\mathbf{or}}~\newline$[\textsl{shuffList}[2i], \textsl{shuffList}[1+2i]]$~\ensuremath{\mathbf{in}}~\textsl{pairList}}{
			\textbf{break}\;
		}{
			$\textsl{pairList}.\textsc{Append}([\textsl{shuffList}[2i], \textsl{shuffList}[1+2i]])$\;
			\lIf{$i=n-1$}{$\textsl{reshuffle}\leftarrow\false$}
		}
	}
}
\Return{\textsl{pairList}}\;
\end{algorithm}

Since we intend to generate instances corresponding to 
different $k$ for each selected sample of TSPLIB benchmark, 
and the number of requests $n$ gets smaller as $k$ decreases, 
we need to produce a decrementable list of node pairs 
such that we can remove some of the node pairs to correspond to 
smaller $k$, but the remaining part of the node pairs still 
contains all locations in the sample other than that as depot 
(i.e., they still need to be selected at least once). 
Therefore, in \Cref{alg:4.3}, we assume that a sorted list 
of node pairs pops the end elements out of it one by one, 
yet it is always guaranteed that every non-depot node 
is selected at least once. 
We insert the node pair in which both nodes are currently 
selected only once into the frontmost of \textsl{head} in \Cref{alg:4.3:line:7}, 
and then append the node pair in which one of them 
is selected only once into \textsl{head} in \Cref{alg:4.3:line:11}. 
Then we keep inserting the node pair in which 
the sum of the repetitions of the two nodes is currently maximum 
into the frontmost of \textsl{tail} in \Cref{alg:4.3:line:22}. 
Note that such insertions and appends require popping the 
corresponding node pair out of \textsl{pairList} and updating $\mathcal{L}$. 
We merge \textsl{head} and \textsl{tail} to obtain 
a sorted list of node pairs \textsl{sortedPairs} at the end of \Cref{alg:4.3}. 
Finally, we can generate a list containing $n$ requests 
for each selected sample of TSPLIB benchmark, 
as shown in \Cref{alg:4.4}. 

\begin{algorithm}[h!]
\DontPrintSemicolon
\caption{Sort the node pairs by the sum of the repeat times of each node in the node pair}
\label{alg:4.3}
\SetDataSty{texttt}
\SetCommentSty{textrm}
\SetKwComment{tcp}{$\rhd$ }{}
\SetKwData{true}{true}
\SetKwData{false}{false}
\SetKwInOut{inpt}{Input}
\SetKwInOut{init}{Init.}
\inpt{\textsl{repeaTimeList}, \textsl{pairList}}
\init{$\mathcal{L}\leftarrow\textsl{repeaTimeList}$, $\textsl{head}\leftarrow[\,]$, $\textsl{tail}\leftarrow[\,]$, $\textsl{max}\leftarrow 0$, $\textsl{maxIndex}\leftarrow -1$, $\textsl{sortedPairs}\leftarrow[\,]$}
\While{$|\textsl{pairList}|>0$}{
	$i\leftarrow 0$\;
	\While{$i<|\textsl{pairList}|$}{
		\uIf{$\mathcal{L}[\textsl{pairList}[i][0]]=1$~\ensuremath{\mathbf{and}}~$\mathcal{L}[\textsl{pairList}[i][1]]=1$}{
			$\mathcal{L}[\textsl{pairList}[i][0]]\leftarrow\mathcal{L}[\textsl{pairList}[i][0]]-1$\;
			$\mathcal{L}[\textsl{pairList}[i][1]]\leftarrow\mathcal{L}[\textsl{pairList}[i][1]]-1$\;
			$\textsl{head}.\textsc{Insert}(0, \textsl{pairList}.\textsc{Pop}(i))$\;\label{alg:4.3:line:7}
		}\uElseIf{$\mathcal{L}[\textsl{pairList}[i][0]]=1$~\ensuremath{\mathbf{or}}~$\mathcal{L}[\textsl{pairList}[i][1]]=1$}{
			$\mathcal{L}[\textsl{pairList}[i][0]]\leftarrow\mathcal{L}[\textsl{pairList}[i][0]]-1$\;
			$\mathcal{L}[\textsl{pairList}[i][1]]\leftarrow\mathcal{L}[\textsl{pairList}[i][1]]-1$\;
			$\textsl{head}.\textsc{Append}(\textsl{pairList}.\textsc{Pop}(i))$\;\label{alg:4.3:line:11}
		}\lElse{
			$i\leftarrow i+1$
		}
	}
	$\textsl{max}\leftarrow 0$\;
	$\textsl{maxIndex}\leftarrow -1$\;
	\For{$j\leftarrow 0$ \KwTo $|\textsl{pairList}|-1$}{
		\If{$\mathcal{L}[\textsl{pairList}[j][0]]+\mathcal{L}[\textsl{pairList}[j][1]]>\textsl{max}$}{
			$\textsl{max}\leftarrow\mathcal{L}[\textsl{pairList}[j][0]] + \mathcal{L}[\textsl{pairList}[j][1]]$\;
			$\textsl{maxIndex}\leftarrow j$\;
		}
	}
	\If{$\textsl{maxIndex}\neq -1$}{
		$\mathcal{L}[\textsl{pairList}[\textsl{maxIndex}][0]]\leftarrow\mathcal{L}[\textsl{pairList}[\textsl{maxIndex}][0]]-1$\;
		$\mathcal{L}[\textsl{pairList}[\textsl{maxIndex}][1]]\leftarrow\mathcal{L}[\textsl{pairList}[\textsl{maxIndex}][0]]-1$\;
		$\textsl{tail}.\textsc{Insert}(0, \textsl{pairList}.\textsc{Pop}(\textsl{maxIndex}))$\;\label{alg:4.3:line:22}
	}
}
$\textsl{sortedPairs}\leftarrow\textsl{head}.\textsc{Extend}(\textsl{tail})$\;
\Return{\textsl{sortedPairs}}\;
\end{algorithm}

\begin{algorithm}[h!]
\DontPrintSemicolon
\caption{Randomly generate the list of requests}
\label{alg:4.4}
\SetDataSty{texttt}
\SetCommentSty{textrm}
\SetKwComment{tcp}{$\rhd$ }{}
\SetKwData{true}{true}
\SetKwData{false}{false}
\SetKwInOut{inpt}{Input}
\SetKwInOut{init}{Init.}
\inpt{$G=(V,E)$, $n=\round(\frac{k(|V|-1)}{2})$, \textsl{sortedPairs}}
\init{$\textsl{avgDistance}\leftarrow\frac{1}{|V|\cdot|V-1|}\sum_{o\in V}\sum_{d\in V\atop d\neq o}\textsc{Distance}(\arc{od})$, $\textsl{avgVolume}\leftarrow 5$, $\textsl{requestList}\leftarrow [\,]$}
\For{$r\leftarrow 0$ \KwTo $n-1$}{
	$q_r\leftarrow\round(\textsc{RandomUniform}(1,2\times\textsl{avgVolume}-1))$\;
	$w_r\leftarrow\round(2\times\textsl{avgDistance}\times q_r\div\textsl{avgVolume})$\;
	$f(r)\leftarrow\textsl{sortedPairs}[r][0]$\;
	$g(r)\leftarrow\textsl{sortedPairs}[r][1]$\;
	$\textsl{requestList}.\textsc{Append}(\langle w_r, q_r, f(r), g(r)\rangle)$\;
}
\Return{\textsl{requestList}}\;
\end{algorithm}

In addition, we generate data for the three types of trucks recursively 
in the order of maximum load capacity of 25, 20, and 15 
until the number of generated trucks reaches $m$. 
And these three types of trucks correspond to their respective 
cost coefficients for each traversing arc of 1.2, 1, and 0.8. 
For example, for a truck $t$ of load capacity 15, 
$l^t_{od}$, i.e., the cost of it passing through the arc 
$\arc{od}$ is $0.8\times\textsc{Distance}(\arc{od})$.\footnote{
	$\textsc{Distance}(\arc{od})$ is the \emph{Euclidean distance} between the coordinates of node $o$ and the coordinates of node $d$.
} 
According to the average volume of 5 for each request set in \Cref{alg:4.4}, 
it is expected that each truck can accommodate four requests at the same time. 

We generate a total of $3\times5\times5=75$ PPDSP instances 
for our comparative experiments based on the proposed formulation 
and on the existing formulation, respectively, 
according to the following parameter settings. 
\begin{itemize}
	\item The selected TSPLIB samples are \{\emph{burma14}, \emph{ulysses16}, \emph{ulysses22}\};
	\item The average repetition rate of the same locations, $k\in\{1, 1.5, 2, 2.5, 3\}$;
	\item The number of trucks $m\in\{2, 4, 6, 8, 10\}$.
\end{itemize}

\subsection{Experimental Settings}
\label{sec:4.2}

All experiments are performed on an Apple M1 Pro chip, 
using the Ubuntu 18.04.6 LTS operating system via the Podman virtual machine 
with 19 GB of allocated memory. 
A single CPU core is used for each experiment. 
We implemented the instance generators for both the proposed formulation 
and the existing formulation by using Python 3. 
Each generated problem instance is solved by 
the MIP optimizer--\textsc{Cplex} of version 20.1.0.0 \cite{cplex2020v1210} 
within 3,600 CPU seconds time limit. 
Source code for our experiments is 
available at \url{https://github.com/ReprodSuplem/PPDSP}. 

\subsection{Results}
\label{sec:4.3}

Tables \ref{tab:4.1}--\ref{tab:4.3} show the performance of 
the existing formulation-based method 
and our proposed formulation-based method 
in terms of the number of generated variables (\#Var.), 
the number of generated constraints (\#Con.), 
and the optimal values (Opt.) 
for the various average repetition rates of the same locations ($k$) 
and the different numbers of trucks ($m$), 
corresponding to TSPLIB samples \emph{burma14}, \emph{ulysses16} and 
\emph{ulysses22}, respectively. 
Each cell recording the left and right values 
corresponds to a comparison item, 
where the left value refers to the performance of 
the item based on the existing formulation-based method, 
while the right value corresponds to the performance of 
the item based on the proposed formulation-based method. 
We compare the performance of these two methods 
in such cells and put the values of the dominant side in bold. 

\begin{table}[t]
\caption{Comparison of the existing formulation-based method with the proposed formulation-based method for TSPLIB sample \emph{burma14} in terms of the number of generated variables, the number of generated constraints, and the optimal values.}
\label{tab:4.1}
\centering
\makebox[\textwidth]{
\setlength{\tabcolsep}{4pt}{
\!\begin{tabular}{c|c|cc|cc|cc|cc|cc}
\toprule
\multirow{2}{*}{$m$} & \multirow{2}{*}{item} & \multicolumn{2}{c|}{$k=1$} & \multicolumn{2}{c|}{$k=1.5$} & \multicolumn{2}{c|}{$k=2$} & \multicolumn{2}{c|}{$k=2.5$} & \multicolumn{2}{c}{$k=3$} \\
 & & \multicolumn{2}{c|}{$(n=7)$} & \multicolumn{2}{c|}{$(n=10)$} & \multicolumn{2}{c|}{$(n=13)$} & \multicolumn{2}{c|}{$(n=16)$} & \multicolumn{2}{c}{$(n=20)$} \\
\midrule
\multirow{3}{*}{$2$} & \#Var. & $576$ & $\bm{458}$ & $1056$ & $\bm{464}$ & $1680$   & $\bm{470}$ & $2448$ & $\bm{476}$ & $3696$  & $\bm{484}$ \\
& \#Con. & $\bm{1027}$ & $1041$ & $1942$ & $\bm{1062}$ & $3145$  & $\bm{1083}$ & $4636$ & $\bm{1104}$ & $7072$ & $\bm{1132}$ \\
& Opt. & $30$      & $30$ & $\bm{42}$ & $40$ & $\bm{56}$ & $53$ & $65$ & $\bm{75}$ & $100$ & $\bm{107}$ \\
\midrule
\multirow{3}{*}{$4$} & \#Var. & $1152$ & $\bm{916}$ & $2112$ & $\bm{928}$ & $3360$   & $\bm{940}$ & $4896$ & $\bm{952}$ & $7392$  & $\bm{968}$ \\
& \#Con. & $\bm{2047}$ & $2075$ & $3874$ & $\bm{2114}$ & $6277$  & $\bm{2153}$ & $9256$ & $\bm{2192}$ & $14124$ & $\bm{2244}$ \\
& Opt. & $\bm{33}$ & $30$ & $\bm{49}$ & $44$ & $58$ & $\bm{61}$ & $53$ & $\bm{77}$ & $68$  & $\bm{103}$ \\
\midrule
\multirow{3}{*}{$6$} & \#Var. & $1728$ & $\bm{1374}$ & $3168$ & $\bm{1392}$ & $5040$   & $\bm{1410}$ & $7344$ & $\bm{1428}$ & $11088$  & $\bm{1452}$ \\
& \#Con. & $\bm{3067}$ & $3109$ & $5806$ & $\bm{3166}$ & $9409$  & $\bm{3223}$ & $13876$ & $\bm{3280}$ & $21176$ & $\bm{3356}$ \\
& Opt. & $\bm{33}$ & $31$ & $\bm{46}$ & $42$ & $57$ & $\bm{58}$ & $43$ & $\bm{81}$ & $46$  & $\bm{103}$ \\
\midrule
\multirow{3}{*}{$8$} & \#Var. & $2304$ & $\bm{1832}$ & $4224$ & $\bm{1856}$ & $6720$   & $\bm{1880}$ & $9792$ & $\bm{1904}$ & $14784$  & $\bm{1936}$ \\
& \#Con. & $\bm{4087}$ & $4143$ & $7738$ & $\bm{4218}$ & $12541$  & $\bm{4293}$ & $18496$ & $\bm{4368}$ & $28228$ & $\bm{4468}$ \\
& Opt. & $\bm{33}$ & $31$ & $\bm{46}$ & $44$ & $56$ & $56$      & $4$  & $\bm{79}$ & $81$  & $\bm{104}$ \\
\midrule
\multirow{3}{*}{$10$} & \#Var. & $2880$ & $\bm{2290}$ & $5280$ & $\bm{2320}$ & $8400$   & $\bm{2350}$ & $12240$ & $\bm{2380}$ & $18480$  & $\bm{2420}$ \\
& \#Con. & $\bm{5107}$ & $5177$ & $9670$ & $\bm{5270}$ & $15673$  & $\bm{5363}$ & $23116$ & $\bm{5456}$ & $35280$ & $\bm{5580}$ \\
& Opt. & $\bm{33}$ & $31$ & $44$      & $44$ & $57$ & $\bm{58}$ & $24$ & $\bm{80}$ & $49$  & $\bm{96}$  \\
\bottomrule
\end{tabular}
}
}
\end{table}

We can see that either the number of variables or 
the number of constraints generated by our proposed method 
is proportional to $m$, 
while neither the number of variables nor the number of 
constraints generated by our proposed method increases significantly 
as $k$ becomes larger. 
Such a result is consistent with \Cref{thm:3}. 
In contrast, although the number of variables and 
the number of constraints produced by the existing method 
is also proportional to $m$, with $k$ becoming larger, 
exponential increases in both of them are observed. 
Furthermore, it is interesting to note that even when $k=1$, 
the proposed method generates fewer variables than the existing method. 

\begin{table}[t]
\caption{Comparison of the existing formulation-based method with the proposed formulation-based method for TSPLIB sample \emph{ulysses16} in terms of the number of generated variables, the number of generated constraints, and the optimal values.}
\label{tab:4.2}
\centering
\makebox[\textwidth]{
\setlength{\tabcolsep}{4pt}{
\!\begin{tabular}{c|c|cc|cc|cc|cc|cc}
\toprule
\multirow{2}{*}{$m$} & \multirow{2}{*}{item} & \multicolumn{2}{c|}{$k=1$} & \multicolumn{2}{c|}{$k=1.5$} & \multicolumn{2}{c|}{$k=2$} & \multicolumn{2}{c|}{$k=2.5$} & \multicolumn{2}{c}{$k=3$} \\
 & & \multicolumn{2}{c|}{$(n=8)$} & \multicolumn{2}{c|}{$(n=11)$} & \multicolumn{2}{c|}{$(n=15)$} & \multicolumn{2}{c|}{$(n=19)$} & \multicolumn{2}{c}{$(n=23)$} \\
\midrule
\multirow{3}{*}{$2$} & \#Var. & $720$ & $\bm{588}$ & $1248$ & $\bm{594}$ & $2176$   & $\bm{602}$ & $3360$ & $\bm{610}$ & $4800$  & $\bm{618}$ \\
& \#Con. & $\bm{1300}$ & $1380$ & $2311$ & $\bm{1401}$ & $4107$  & $\bm{1429}$ & $6415$ & $\bm{1457}$ & $9235$ & $\bm{1485}$ \\
& Opt. & $94$      & $94$ & $\bm{96}$  & $94$  & $109$ & $\bm{112}$ & $157$ & $\bm{187}$ & $168$ & $\bm{230}$ \\
\midrule
\multirow{3}{*}{$4$} & \#Var. & $1440$ & $\bm{1176}$ & $2496$ & $\bm{1188}$ & $4352$   & $\bm{1204}$ & $6720$ & $\bm{1220}$ & $9600$  & $\bm{1236}$ \\
& \#Con. & $\bm{2592}$ & $2752$ & $4611$ & $\bm{2791}$ & $8199$  & $\bm{2843}$ & $12811$ & $\bm{2895}$ & $18447$ & $\bm{2947}$ \\
& Opt. & $\bm{99}$ & $98$ & $\bm{104}$ & $101$ & $104$ & $\bm{124}$ & $140$ & $\bm{198}$ & $29$  & $\bm{247}$ \\
\midrule
\multirow{3}{*}{$6$} & \#Var. & $2160$ & $\bm{1764}$ & $3744$ & $\bm{1782}$ & $6528$   & $\bm{1806}$ & $10080$ & $\bm{1830}$ & $14400$  & $\bm{1854}$ \\
& \#Con. & $\bm{3884}$ & $4124$ & $6911$ & $\bm{4181}$ & $12291$  & $\bm{4257}$ & $19207$ & $\bm{4333}$ & $27659$ & $\bm{4409}$ \\
& Opt. & $\bm{99}$ & $98$ & $\bm{105}$ & $101$ & $104$ & $\bm{139}$ & $92$  & $\bm{205}$ & $0$   & $\bm{226}$ \\
\midrule
\multirow{3}{*}{$8$} & \#Var. & $2880$ & $\bm{2352}$ & $4992$ & $\bm{2376}$ & $8704$   & $\bm{2408}$ & $13440$ & $\bm{2440}$ & $19200$  & $\bm{2472}$ \\
& \#Con. & $\bm{5176}$ & $5496$ & $9211$ & $\bm{5571}$ & $16383$  & $\bm{5671}$ & $25603$ & $\bm{5771}$ & $36871$ & $\bm{5871}$ \\
& Opt. & $\bm{99}$ & $98$ & $91$ & $\bm{98}$   & $89$  & $\bm{138}$ & $75$  & $\bm{176}$ & $96$  & $\bm{241}$ \\
\midrule
\multirow{3}{*}{$10$} & \#Var. & $3600$ & $\bm{2940}$ & $6240$ & $\bm{2970}$ & $10880$   & $\bm{3010}$ & $16800$ & $\bm{3050}$ & $24000$  & $\bm{3090}$ \\
& \#Con. & $\bm{6468}$ & $6868$ & $11511$ & $\bm{6961}$ & $20475$  & $\bm{7085}$ & $31999$ & $\bm{7209}$ & $46083$ & $\bm{7333}$ \\
& Opt. & $\bm{99}$ & $98$ & $96$ & $\bm{101}$  & $111$ & $\bm{130}$ & $82$  & $\bm{186}$ & $27$  & $\bm{221}$  \\
\bottomrule
\end{tabular}
}
}
\end{table}

As $k$ increases, the optimal value obtained by 
the proposed method increases within the time limit; 
while there is no such trend in the optimal values 
obtained by the existing method. 
In addition, as $m$ grows up, in theory, 
the upper bound of the optimal value cannot be smaller. 
However, the optimal values obtained by the proposed method 
and the existing method do not maintain the trend of increasing, 
instead, they both have inflection points. 
This is due to the fact that the search space of the problem becomes huge, 
which makes it inefficient for the solver to update the optimized solution. 
There are even cases on the existing method where the initial solution 
is not obtained until the end of time. 

\begin{table}[t]
\caption{Comparison of the existing formulation-based method with the proposed formulation-based method for TSPLIB sample \emph{ulysses22} in terms of the number of generated variables, the number of generated constraints, and the optimal values.}
\label{tab:4.3}
\centering
\makebox[\textwidth]{
\setlength{\tabcolsep}{4pt}{
\!\begin{tabular}{c|c|cc|cc|cc|cc|cc}
\toprule
\multirow{2}{*}{$m$} & \multirow{2}{*}{item} & \multicolumn{2}{c|}{$k=1$} & \multicolumn{2}{c|}{$k=1.5$} & \multicolumn{2}{c|}{$k=2$} & \multicolumn{2}{c|}{$k=2.5$} & \multicolumn{2}{c}{$k=3$} \\
 & & \multicolumn{2}{c|}{$(n=11)$} & \multicolumn{2}{c|}{$(n=16)$} & \multicolumn{2}{c|}{$(n=21)$} & \multicolumn{2}{c|}{$(n=26)$} & \multicolumn{2}{c}{$(n=32)$} \\
\midrule
\multirow{3}{*}{$2$} & \#Var. & $1248$ & $\bm{1074}$ & $2448$ & $\bm{1084}$ & $4048$   & $\bm{1094}$ & $6048$ & $\bm{1104}$ & $8976$  & $\bm{1116}$ \\
& \#Con. & $\bm{2311}$ & $2685$ & $4636$ & $\bm{2720}$ & $7761$  & $\bm{2755}$ & $11686$ & $\bm{2790}$ & $17452$ & $\bm{2832}$ \\
& Opt. & $116$ & $\bm{120}$ & $161$ & $\bm{206}$ & $182$ & $\bm{225}$ & $85$ & $\bm{297}$ & $79$ & $\bm{336}$ \\
\midrule
\multirow{3}{*}{$4$} & \#Var. & $2496$ & $\bm{2148}$ & $4896$ & $\bm{2168}$ & $8096$   & $\bm{2188}$ & $12096$ & $\bm{2208}$ & $17952$  & $\bm{2232}$ \\
& \#Con. & $\bm{4611}$ & $5359$ & $9256$ & $\bm{5424}$ & $15501$  & $\bm{5489}$ & $23346$ & $\bm{5554}$ & $34872$ & $\bm{5632}$ \\
& Opt. & $\bm{134}$ & $120$ & $141$ & $\bm{185}$ & $53$  & $\bm{216}$ & $34$ & $\bm{277}$ & $15$ & $\bm{327}$ \\
\midrule
\multirow{3}{*}{$6$} & \#Var. & $3744$ & $\bm{3222}$ & $7344$ & $\bm{3252}$ & $12144$   & $\bm{3282}$ & $18144$ & $\bm{3312}$ & $26928$  & $\bm{3348}$ \\
& \#Con. & $\bm{6911}$ & $8033$ & $13876$ & $\bm{8128}$ & $23241$  & $\bm{8223}$ & $35006$ & $\bm{8318}$ & $52292$ & $\bm{8432}$ \\
& Opt. & $115$ & $\bm{125}$ & $136$ & $\bm{192}$ & $0$   & $\bm{227}$ & $45$ & $\bm{242}$ & $0$  & $\bm{343}$ \\
\midrule
\multirow{3}{*}{$8$} & \#Var. & $4992$ & $\bm{4296}$ & $9792$ & $\bm{4336}$ & $16192$   & $\bm{4376}$ & $24192$ & $\bm{4416}$ & $35904$  & $\bm{4464}$ \\
& \#Con. & $\bm{9211}$ & $10707$ & $18496$ & $\bm{10832}$ & $30981$  & $\bm{10957}$ & $46666$ & $\bm{11082}$ & $69712$ & $\bm{11232}$ \\
& Opt. & $\bm{124}$ & $117$ & $102$ & $\bm{206}$ & $25$  & $\bm{224}$ & $58$ & $\bm{228}$ & $41$ & $\bm{276}$ \\
\midrule
\multirow{3}{*}{$10$} & \#Var. & $6240$ & $\bm{5370}$ & $12240$ & $\bm{5420}$ & $20240$   & $\bm{5470}$ & $30240$ & $\bm{5520}$ & $44880$  & $\bm{5580}$ \\
& \#Con. & $\bm{11511}$ & $13381$ & $23116$ & $\bm{13536}$ & $38721$  & $\bm{13691}$ & $58326$ & $\bm{13846}$ & $87132$ & $\bm{14032}$ \\
& Opt. & $\bm{124}$ & $123$ & $31$  & $\bm{173}$ & $44$  & $\bm{204}$ & $42$ & $\bm{243}$ & $44$  & $\bm{319}$  \\
\bottomrule
\end{tabular}
}
}
\end{table}

Last but not least, we can clearly see that 
for those problems with larger $k$ 
(i.e., more nodes that can be integrated), 
the method based on our proposed formulation 
generates fewer variables and constraints, 
as well as achieves larger optimal values, 
than the method based on the existing formulation.

\section{Conclusion}
\label{sec:5}

In this paper, we revisit PPDSP on road networks with the integratable nodes. 
For such application scenarios, we define a location-based graph-theoretic model 
and give the corresponding MIP formulation. 
We prove the correctness of this formulation as well as analyze its space complexity. 
We compare the proposed method with the existing formulation-based method. 
The experimental results show that, for the instances with more integratable nodes, 
our method has a significant advantage over the existing method 
in terms of the generated problem size, and the optimized values.

\section*{Acknowledgments}

The authors are grateful for the support of 
the Progressive Logistic Science Corporate Sponsored Research Programs of 
the Research Center for Advanced Science and Technology, 
The University of Tokyo.

\clearpage
\bibliographystyle{apalike-ejor}
\biboptions{authoryear}
\bibliography{bibsample}

\clearpage
\appendix
\def\appendixname{}
\section{Existing Formulation}
\label{apdx:1}

Here we provide a detailed description of the existing formulation 
proposed in \cite{journals/eor/QiuFN17} that were used to perform 
the comparative experiments in \Cref{sec:4}. 

As we stated before, in order to compare the most essential differences 
between the request-based model and the location-based model, 
the constraints on time windows are not considered in this study. 
However, because in the existing formulation, the time windows constraints 
also act as the subtour elimination constraint. 
For this reason, we replace those constraints on time windows 
in the existing formulation with the same MTZ subtour elimination constraint 
as in our proposed formulation to ensure the fairness of the comparative experiments. 

In the existing formulation, $n$ requests are represented as a 
directed graph $G=(V,E)$, whose node set $V=\{0\}\cup P\cup D\cup\{2n+1\}$, 
where $P=\{1,2,\ldots,n\}$ is the set of pickup nodes, 
$D=\{n+1,n+2,\ldots,2n\}$ is the set of dropoff nodes, 
and node $0$ and $2n+1$ indicate two depots.\footnote{
In the experiments of this study, such two depots are in the identical location.} 
A request $r$ is formalized by using the combination $(r,r+n)$ 
to refer to the pickup and dropoff nodes, and each node $v~(v\in P\cup D)$ 
corresponds to a change amount $q_v$ of the truck's load, 
where $q_r$ and $q_{r+n}$ respectively correspond to the loading and unloading 
of request $r$, $q_r>0$, $q_{r+n}<0$, $q_r=|q_{r+n}|$ and $r\in\{1,2,\ldots,n\}$. 
In addition, we particularly specify $q_0=q_{2n+1}=0$ for the two depot nodes. 
The definitions and notations of the other variables are consistent with 
that previously described in this paper. 
The existing formulation used for comparative experiments are as follows. 

\begin{align}
\max~&\sum_{t\in T}\bigg(\sum_{o\in P}w_{o}\sum_{d\in V}x^t_{od}\bigg)-\sum_{t\in T}\sum_{o\in V}\sum_{d\in V}(l^t_{od}\cdot x^t_{od}),\label{eq:a.1}\\
\textrm{s.t.}~&x^t_{od}\in\{0,1\},&\hspace{-85pt}\forall (t,o,d):t\in T, o,d\in V,\label{eq:a.2}\\
&\sum_{d\in V}x^t_{0d}=\sum_{o\in V}x^t_{o(2n+1)}=1,&\hspace{-85pt}\forall t:t\in T,\label{eq:a.3}\\
&\sum_{t\in T}\sum_{o\in V}x^t_{od}\le 1,&\hspace{-85pt}\forall (t,d):t\in T, d\in P,\label{eq:a.4}\\
&\sum_{d\in V}x^t_{od}-\sum_{d\in V}x^t_{(o+n)d}=0,&\hspace{-85pt}\forall (t,o):t\in T, o\in P,\label{eq:a.5}\\
&\sum_{d\in V}x^t_{vd}-\sum_{o\in V}x^t_{ov}=0,&\hspace{-85pt}\forall (t,v):t\in T, v\in P\cup D,\label{eq:a.6}\\
&u^t_d-u^t_o\ge 1-|V|(1-x^t_{od}),&\hspace{-85pt}\forall (t,o,d):t\in T, o,d\in V, o\neq d,\label{eq:a.7}\\
&u^t_d-u^t_o>0,&\hspace{-85pt}\forall (t,o,d):t\in T, o\in P, d\in D,\label{eq:a.8}\\
&h^t_d-h^t_o\ge q_d-c^t(1-x^t_{od}),&\hspace{-85pt}\forall (t,o,d):t\in T, o,d\in V,\label{eq:a.9}\\
&\max\{0,q_v\}\le h^t_v\le\min\{c^t,c^t+q_v\},&\hspace{-85pt}\forall (t,v):t\in T, v\in V,\label{eq:a.10}\\
&0\le u^t_v\le |V|-1,&\hspace{-85pt}\forall (t,v):t\in T, v\in V\setminus\{0\},\label{eq:a.11}\\
&x^t_{vv}=0,&\hspace{-85pt}\forall (t,v):t\in T, v\in V,\label{eq:a.12}\\
&x^t_{0d}=0,&\hspace{-85pt}\forall (t,d):t\in T, d\in D,\label{eq:a.13}\\
&x^t_{o(2n+1)}=0,&\hspace{-85pt}\forall (t,o):t\in T, o\in P,\label{eq:a.14}\\
&x^t_{o0}=0,&\hspace{-85pt}\forall (t,o):t\in T, o\in P\cup D,\label{eq:a.15}\\
&x^t_{(2n+1)d}=0,&\hspace{-85pt}\forall (t,d):t\in T, d\in P\cup D.\label{eq:a.16}
\end{align}

Eq. \eqref{eq:a.3} restricts that all routes start and end at the depots, 
whereas Eq. \eqref{eq:a.4} constrains each node $d$ cannot be visited more than once. 
Eq. \eqref{eq:a.5} guarantees that pickup and dropoff from a request 
can only be assigned to the same truck. 
Eq. \eqref{eq:a.6} ensures that, for each non-depot node, 
the number of visits and the number of departures of a truck must be equal. 
Eqs. \eqref{eq:a.7} and \eqref{eq:a.8}, associated with the domain of 
$u^t_v$ shown in Eq. \eqref{eq:a.11}, are the MTZ subtour elimination constraint 
and the \emph{loading before unloading} constraint. 
Eqs. \eqref{eq:a.9} and \eqref{eq:a.10} correspond to the capacity constraint. 
Eqs. \eqref{eq:a.12}--\eqref{eq:a.16} are the constraints in 
\cite{journals/eor/QiuFN17} described as helping to improve the solving speed. 
Eq. \eqref{eq:a.12} blocks cycles at the nodes, 
while Eq. Eq. \eqref{eq:a.13} forbids the direct visit of 
dropoff nodes from the start depot. 
Eq. \eqref{eq:a.14} (resp. Eq. \eqref{eq:a.15}) 
prevents an immediate access to the end depot 
after visiting a pickup node (resp. visiting the start depot again). 
Eq. \eqref{eq:a.16} avoids trucks from starting at the end depot. 

It should be noted that we found three incorrectnesses in 
the existing formulations in \cite{journals/eor/QiuFN17}, 
which correspond to Eqs. \eqref{eq:a.3}, \eqref{eq:a.4} and \eqref{eq:a.12} 
in the above formulation, and we have rectified each of them.

\end{document}